\documentclass{article}
\usepackage[letterpaper, margin=1in]{geometry}

\usepackage[sort]{natbib}

\usepackage{graphicx} 
\usepackage{algpseudocode}
\usepackage{ragged2e}
\usepackage{hyperref}
\usepackage{listings}
\usepackage{amsmath}
\usepackage{amssymb}
\usepackage{amsfonts}
\usepackage{amsthm}
\usepackage[export]{adjustbox}
\usepackage{mdframed}
\usepackage{mathtools}
\usepackage{hyperref}
\usepackage{xcolor}
\usepackage{soul}
\usepackage{tikz}
\usepackage{cleveref}
\usepackage[utf8]{inputenc}
\usepackage[T1]{fontenc}
\usepackage{thmtools}
\usepackage{newunicodechar}

\newtheorem{theorem}{Theorem}
\newtheorem{lemma}[theorem]{Lemma}
\newtheorem{corollary}[theorem]{Corollary}
\newtheorem{claim}[theorem]{Claim}
\newtheorem{definition}[theorem]{Definition}

\usepackage{tcolorbox}
\tcbuselibrary{skins,breakable}

\newtcolorbox{AlgDescBox}[1]{%
  enhanced,
  breakable,
  colback=gray!10,
  colframe=gray!10,   
  boxrule=0pt,        
  arc=0pt,            
  title=\textbf{#1},
  fonttitle=\bfseries,
  coltitle=black,
  colbacktitle=gray!10, 
  left=1em,
  right=1em,
  top=0.5em,
  bottom=0.5em,
}

\newcommand{\iref}[1]{(\ref{#1})}

\usepackage[boxed,linesnumbered,vlined]{algorithm2e}
\DontPrintSemicolon
\SetKwInput{SharedData}{shared data}
\SetKwInput{LocalData}{local data}
\SetKw{Break}{break}
\SetKw{KwProcedure}{procedure}
\SetKw{True}{true}
\SetKw{False}{false}
\SetKw{KwAnd}{and}
\SetKw{KwOr}{or}
\SetKw{Not}{not}
\SetKwIF{If}{ElseIf}{Else}{if}{then}{else if}{else}{end if}
\SetKwFor{For}{for}{do}{end for}
\SetKwFor{While}{while}{do}{end while}
\SetKw{DownTo}{down to}

\SetKwIF{WithProbability}{ElseIf}{Else}{with probability}{do}{else if}{else}{end}

\SetKwBlock{Initially}{initially do}{end}
\SetKwFor{UponReceiving}{upon receiving}{do}{end}
\SetKwFor{AsSoonAs}{as soon as}{do}{end}
\SetKwFor{Foreach}{for each}{do}{end}
\SetKw{From}{from}
\SetKwBlock{Atomically}{atomically do}{end}

\SetKwFor{Procedure}{procedure}{}{end procedure}

\newcommand{\newData}[2]{\newcommand{#1}{\DataSty{#2}\xspace}}
\newcommand{\newFunc}[2]{\newcommand{#1}{\FuncSty{#2}\xspace}}

\newData{\Input}{input}

\newFunc{\Read}{read}
\newFunc{\Write}{write}

\newcommand{\Set}[1]{\left\{{#1}\right\}}

\newcommand{\Tuple}[1]{\left\langle{#1}\right\rangle}

\newunicodechar{¬}{\neg}   
\newunicodechar{∧}{\wedge} 
\newunicodechar{⊕}{\oplus} 
\newunicodechar{∨}{\vee}   
\newunicodechar{→}{\rightarrow}  
\newunicodechar{⇒}{\Rightarrow}  
\newunicodechar{←}{\leftarrow}   
\newunicodechar{↔}{\leftrightarrow}   
\newunicodechar{⇔}{\Leftrightarrow}   
\newunicodechar{⊢}{\vdash}   

\newunicodechar{⊗}{\otimes} 

\newunicodechar{↦}{\mapsto}  
\newunicodechar{↪}{\hookrightarrow}  
\newunicodechar{↣}{\rightarrowtail}  

\newunicodechar{∀}{\forall} 
\newunicodechar{∃}{\exists} 

\newunicodechar{∆}{\Delta} 
\newunicodechar{∇}{\nabla} 

\newunicodechar{∩}{\cap}    
\newunicodechar{∪}{\cup}    
\newunicodechar{∖}{\setminus}    
\newunicodechar{×}{\times}    
\newunicodechar{∈}{\in}     
\newunicodechar{∉}{\not\in}     
\newunicodechar{∋}{\ni}     
\newunicodechar{⊆}{\subseteq}    
\newunicodechar{⊇}{\supseteq}    
\newunicodechar{⊑}{\sqsubseteq}    
\newunicodechar{⊒}{\sqsupseteq}    
\newunicodechar{≾}{\lesssim}
\newunicodechar{≿}{\grtsim}
\newunicodechar{⪯}{\preceq}
\newunicodechar{⪰}{\succeq}

\newunicodechar{∘}{\circ}    
\newunicodechar{⋅}{\cdot}    

\newunicodechar{∅}{\emptyset}  
\DeclareUnicodeCharacter{2115}{\mathbb{N}} 
\DeclareUnicodeCharacter{2124}{\mathbb{Z}} 
\DeclareUnicodeCharacter{211A}{\mathbb{Q}} 
\DeclareUnicodeCharacter{211D}{\mathbb{R}} 
\DeclareUnicodeCharacter{2102}{\mathbb{C}} 
\newunicodechar{ℵ}{\aleph}     
\newunicodechar{ℶ}{\beth}      

\DeclareUnicodeCharacter{2131}{\mathcal{F}} 

\newunicodechar{∑}{\sum}  
\newunicodechar{∏}{\prod} 
\newunicodechar{∫}{\int}  
\newunicodechar{±}{\pm}   

\newunicodechar{≠}{\neq}    
\newunicodechar{≡}{\equiv}  
\newunicodechar{≈}{\approx} 
\newunicodechar{≃}{\simeq}  
\newunicodechar{≅}{\cong}   
\newunicodechar{≤}{\le}     
\newunicodechar{≥}{\ge}     
\newunicodechar{∣}{\divides}

\newunicodechar{⊤}{\top}    
\newunicodechar{⊥}{\bot}    
\newunicodechar{∞}{\infty}  
\newunicodechar{§}{\S}      
\newunicodechar{□}{\Box}    
\newunicodechar{◇}{\Diamond}      

\newunicodechar{Γ}{\Gamma}  
\newunicodechar{Δ}{\Delta}  
\newunicodechar{Θ}{\Theta}  
\newunicodechar{θ}{\theta}  
\newunicodechar{Λ}{\Lambda} 
\newunicodechar{Ξ}{\Xi}     
\newunicodechar{Π}{\Pi}     
\newunicodechar{Σ}{\Sigma}  
\newunicodechar{Φ}{\Phi}    
\newunicodechar{Ω}{\Omega}  
\newunicodechar{α}{\alpha}  
\newunicodechar{β}{\beta}   
\newunicodechar{γ}{\gamma}  
\newunicodechar{δ}{\delta}  
\newunicodechar{ε}{\epsilon} 
\newunicodechar{ζ}{\zeta}   
\newunicodechar{η}{\eta}    
\newunicodechar{ι}{\iota}   
\newunicodechar{κ}{\kappa}  
\newunicodechar{λ}{\lambda} 
\newunicodechar{μ}{\mu}     
\newunicodechar{ξ}{\xi}     
\newunicodechar{π}{\pi}     
\newunicodechar{ρ}{\rho}    
\newunicodechar{σ}{\sigma}  
\newunicodechar{τ}{\tau}    
\newunicodechar{φ}{\phi}    
\newunicodechar{χ}{\chi}    
\newunicodechar{ψ}{\psi}    
\newunicodechar{ω}{\omega}  

\newcommand{\concept}[1]{\textbf{#1}}
\DeclarePairedDelimiter{\abs}{\lvert}{\rvert}

\newData{\Adopt}{adopt}
\newData{\Commit}{commit}
\newData{\Nonce}{nonce}

\title{Obfuscated Consensus}
\author{
        James Aspnes\thanks{Yale University} \and
        Shlomi Dolev\thanks{Ben-Gurion University of the Negev} \and
        Amit Hendin\footnotemark[2]
    }

\date{}

\begin{document}

\maketitle

\begin{abstract}
    The classic Fischer, Lynch, and Paterson impossibility
    proof~\cite{FischerLP1985,LouiA1987} demonstrates that any
    deterministic protocol for consensus in either a message-passing
    or shared-memory system must violate at least one of termination,
    validity, or agreement in some execution. But it does not provide
    an efficient procedure to find such a bad execution.

    We show that for wait-free shared memory consensus, given a
    protocol in which each process performs at most $s$ steps computed
    with total time complexity at most $t$, there exists an adversary
    algorithm that takes the process's programs as input and computes
    within $O(st)$ time a schedule that violates agreement. We argue
    that this bound is tight assuming the random oracle hypothesis:
    there exists a deterministic \concept{obfuscated consensus
    protocol} that forces the adversary to spend $Ω(st)$ time to find
    a bad execution despite having full access to all information
    available to the protocol.

    This bound is based on a general reduction from constructing an obfuscated consensus protocol to constructing an \concept{obfuscated threshold function} that
    provably costs $Ω(t)$ time to evaluate on a single input, where
    $t$ is a tunable parameter, and for which an adversary with access
    to the threshold function implementation cannot extract the
    threshold any faster than by doing binary search. We give a
    particular implementation of such an obfuscated threshold function
    that is not very efficient but that is provably secure assuming
    the random oracle hypothesis. Since our obfuscated consensus
    protocol does not depend on the specific details of this
    construction, it may be possible to replace it with one that is
    more efficient or requires weaker cryptographic assumptions, a
    task we leave for future work.
\end{abstract}

\section{Introduction}
Asynchronous consensus is one of the most investigated tasks in distributed computing. Consensus over distributed inputs allows a reduction to centralized computation, abstracting away the inherent coordination difficulties that distributed systems cope with. 

Asynchronous consensus is a building block in structuring replicated
state machines and efficient proof of authority/stake Blockchains, e.g., \cite{DBLP:journals/tdsc/DolevGNW24}.

Unfortunately, the celebrated result of Fischer, Lynch, and Paterson
\cite{FischerLP1985} proves that there is no deterministic algorithm
that can yield an asynchronous consensus. Subsequent consensus
protocols have escaped this bound by adjusting the requirements of the
problem.
Paxos \cite{DBLP:journals/tocs/Lamport98} succeeded in ensuring the consistency (all participants decide on the same value) and validity (the decided value appears in at least one of the inputs) of the distributed decision but sacrifices the liveness as an adversarial scheduler can keep the system undecided forever. 
Alternatively, randomization (for example,
\cite{DBLP:conf/podc/Ben-Or83}) allows consensus to be solved at the
cost of replacing a deterministic termination guarantee with a
probabilistic one, where the protocol finishes with probability $1$ in
the limit. In these protocols, the adversary scheduler that could
otherwise prevent agreement following the procedure given by the FLP
proof is unable to predict the future and thus falls off the bad path.

We investigate an alternative approach, where validity and termination
are guaranteed in all executions, and agreement holds for all but a
small number of executions that are computationally expensive to
compute, even for an adversary given access to the entire code and
knowledge of the processes in the protocol. This idea is inspired by
the Fiat-Shamir heuristic~\cite{FS} for extracting a zero-knowledge
challenge from a hash of public values assuming a random oracle. In
our approach, the scheduling choices of the adversary in an
asynchronous shared-memory system, as observed by the processes,
become the input to an obfuscated threshold function. Each process
computes its output by computing a single value of this function,
while the adversary can only prevent agreement through the much more
expensive process of finding the threshold input at which the value of
the function changes from $0$ to $1$. This gap can in principle be
exploited to obtain a fully deterministic wait-free asynchronous
protocol that solves consensus in most executions, assuming the
adversary has limited computational power.

This approach has some limitations. We show that, given any protocol
in which each process performs $s$ steps and computes for $t$ time,
there is an adversary program that outputs an agreement-violating
schedule in $O(st)$ time in the worst case. 

Conversely, we provide a reduction from our obfuscated consensus
problem to constructing an obfuscated threshold function, and give a
particular implementation of such a function as an obfuscated truth
table. Here each bit of the truth table is hidden using a post-quantum time
lock puzzle (see, e.g.,
\cite{rivest1996time,DBLP:conf/asiacrypt/AfsharCHLM23} and the
references therein) for efficiently structuring obfuscated programs
(see, e.g., \cite{DBLP:journals/cacm/JainLS24} and the references
therein). This allows us to tune the base cost of computing the
threshold function for both the processes and the adversary, and to
prevent the adversary from finding the hidden threshold any faster
than by binary search over a range exponential in the step complexity
of the consensus protocol.

While our particular implementation of an obfuscated threshold
function involves a relatively expensive setup phase, it is possible
to imagine an implementation using more sophisticated cryptographic
techniques that would avoid this overhead. We discuss this possibility
further in \Cref{section-conclusion}.

Before moving on, we would like to discuss why investigating such adversarial settings is interesting in the first place. 
In the asynchronous setting, the purpose of the adversarial scheduler is to theoretically induce the worst-case scenario on the protocol; if a protocol works under such settings, then no schedule can cause it to violate its correctness conditions or loop indefinitely. 
In randomized protocols, theoretically, the adversary cannot guess the next configuration with certainty and thus cannot know the scheduling that will keep the system in a bad state. Solely based on the theoretical implications the motivation for a computationally constraint adversary naturally arises: if we can hide the next configuration behind a computational puzzle that is too difficult for the adversary, it cannot not find the bad schedule with certainty. 
The real question arises when looking at the practical implications of this. For protocols that guarantee agreement and validity, bad schedules are infinite loops, timing messages or scheduling processes in just the right order such that the system forever remains in an indecisive state. In the randomized case, if the protocol enters a loop of indecisive configurations, then the randomized step can pull it out by randomly moving the system to a different configuration and breaking the loop. The question remains: what is the practical implication of the computationally bounded adversary? Resilience against an actual adversary. Rather than attempt to pull the protocol out of non-terminating loops, we defend the protocol against real scheduling attacks such as Adversarial Transaction Reordering, that occur when miners, validators, or automated bots manipulate the sequence of pending transactions in a block to extract financial profit. By exploiting the time delay between when a transaction is broadcast and when it is confirmed, bad actors engage in front-running, back-running, and sandwich attacks at the expense of regular users. By obfuscating the keys parts of the protocol behind a cryptographic puzzle, we can create a time delta between the honest participants and the malicious ones, preventing adversarial attacks on the schedule in cases where the adversary algorithm cannot solve the cryptographic puzzle fast enough.

\paragraph{Connection to blockchains.}
Our scheduler model abstracts schedule attacks that arise in blockchain systems when an adversary influences the order in which transactions, messages, or block proposals are observed and processed.  In this correspondence, processes represent validators or protocol participants, shared memory represents the publicly evolving protocol state, and a schedule represents the adversary's chosen ordering of protocol events.  A bad asynchronous schedule therefore corresponds to a strategically chosen ordering that drives the system into an undesirable execution, such as disagreement, delayed decision, or exploitable transaction reordering.  Obfuscated consensus captures the goal of making such an ordering computationally expensive to find: honest participants execute the protocol directly, while the adversary must search for the damaging schedule.

\subsection{Overview} The system settings are described in the next section.
The problem of \concept{obfuscated consensus} is described in
Section~\ref{section-obfuscated-consensus}: this replaces the usual
agreement condition for consensus that requires all processes output
the same value in each execution with a new \concept{obfuscated
agreement} condition that only requires that it is computationally
expensive for the adversary scheduler to find a schedule that violates
agreement.
Section \ref{section-consensus-from-threshold} gives a
reduction from obfuscated consensus to a purely 
cryptographic problem, that of obfuscating a threshold
function. This reduction has parameters allowing both the
computation cost incurred by each process and the ratio of the
adversary's cost to each process's cost to be adjusted independently.
An implementation of an obfuscated threshold function, along with a
proof of security, is given in Section~\ref{section-obfuscation}.
Finally, we discuss some of the limitations of our current
implementation and possibilities for future work in
Section~\ref{section-conclusion}.

\section{Model}
\label{section-model}

We use a standard asynchronous shared-memory model, with some small
tweaks to include explicitly the programs used by the processes and
adversary to compute their choices.

The system consists of $n$ asynchronous \concept{processes} $p_1,...,p_n$
that communicate by reading and writing to shared objects, with timing
controlled by an \concept{adversary scheduler}.

The behaviour of each process is controlled by a random-access machine
program $P$ that implements a function $Q×V→Q×A$, where $Q$ is the set
of states of the process, $V$ is the set of possible return values of
operations (including a special null value used to indicate the first
step by the process), and $A$ is the next operation to execute.
We consider the cost of computing the transition function a single time to be constant.

The adversary scheduler is a program in the same model that takes as
input a step bound $s$, a sequence of programs $\Tuple{P_1,\dots,P_n}$
and initial states $q_1,\dots,q_n$, and outputs a sequence of
process ids, where each id occurs at most $s$ times, specifying the
order in which the processes take steps. We will refer to this
sequence of process ids as a \concept{schedule}.

Communication is via \concept{shared memory}: a collection of 
\concept{atomic read/write registers} $r_1,\dots,r_m$, where applying
a $\Read$ operation to a register returns its current value and
applying a $\Write$ operation sets the value and returns nothing.

A \concept{configuration} of the system consists of the states of the
processes and registers, together with at most one pending operation
for each process.

Every schedule yields an \concept{execution}, an alternating sequence
of configurations and operations $C_0 α_1, C_1 α_2 C_2 \dots$ in which
each $C_{i+1}$ is the result of applying operation $α_{i+1}$ to
configuration $C_i$. Because the processes are deterministic, the
adversary can simulate the execution corresponding to any particular
schedule in time proportional to the total time used by the processes
to choose their steps in that execution.

We assume that both the processes and the adversary have access to a
\concept{random oracle} $H$, which is an idealized version of the
cryptographic hash function. The random oracle acts as a black box
that can be queried with input $x$ to return output $y$. The box is
consistent; that is, any input $x$ will always output the same $y$.
The box is uniformly random, which means the probability of getting
output $y$ for an input $x$ that was not already queried is uniform.
Though everyone has access to this black box and can query it, its
internal workings are unknown (see Chapter 5 of {\cite{KatzLindell}}).

\section{Obfuscated consensus}
\label{section-obfuscated-consensus}

A consensus protocol is a protocol that is run by a set of $n$
asynchronous processes, each of which starts with a binary input value
$v_i \in \{0,1\}$, runs until it reaches a decision value $d_i
\in\{0,1\}$, then halts. A consensus protocol is correct if it
satisfies the following requirements:
\begin{description}
\item[Agreement] A consensus protocol satisfies agreement if no two
    processes decide on different values: $\forall i,j: d_i=d_j$.
\item[Validity] A consensus protocol is valid if all processes decide
on a value that was the input of some process, $\forall i \exists j: d_i = v_j$. 
\item[Termination] Every non-faulty process decides after a finite number of steps.
\end{description}

We think of the consensus problem as an adversarial game between the
scheduler and the processes. The processes try to reach agreement on the
same value by reading and writing to shared objects, and the scheduler
tries to prevent them from doing so by strategically ordering these operations.

The Fischer-Lynch-Paterson (FLP) impossibility
result~\cite{FischerLP1985}, as extended to shared memory by Loui and
Abu-Amara~\cite{LouiA1987}, shows that in an asynchronous system with
$n≥2$ processes in which at least one process can fail by crashing,
there is always an adversary
strategy that finds an execution in which at least one of the three
conditions is violated. We will be interested in constructing protocols
where finding this strategy is computationally difficult.

Define an \concept{obfuscated consensus
protocol} to be a protocol that satisfies termination and validity
always, but replaces agreement with a condition that we can describe
informally as
\begin{description}
    \item[Obfuscated agreement] The cost to the adversary to compute
        a schedule that violates agreement substantially exceeds the
        cost to the processes to execute the protocol.
\end{description}

More formally, the adversary must solve an \concept{obfuscated
consensus problem}: Given process programs $P_1,\dots,P_n$ and a step
bound $s$, find a schedule in which each process takes $s$ steps that
violates one of termination, validity, or agreement. A successful
obfuscated consensus protocol is one for which solving this problem is
expensive, for some reasonable definition of expensive. We will start
by bounding this cost from above.

\section{Bounding the cost of finding a bad execution}

The core idea of the FLP impossibility proof~\cite{FischerLP1985} is
the notion of \concept{bivalence}.
A configuration $C$ is bivalent if there is a path from it to a
configuration where a process decides on the value $1$, and there
is also a path from it to another configuration where a process decides on the value $0$.
Similarly, a configuration $C$ is univalent if all the paths from $C$ lead to configurations with an identical decision value.
Because of the agreement property, any
configuration in which some process has decided is necessarily
univalent.

The FLP proof, in both its original form~\cite{FischerLP1985}, and its
extension to shared memory by Loui and Abu-Amara \cite{LouiA1987},
gives an explicit procedure for constructing an
infinite non-terminating execution for any protocol that satisfies
agreement and validity, by starting in a bivalent configuration and
always choosing bivalent successor configurations.
However, the computational cost of this construction is not considered.
We would like to find an adversary strategy that allows it to prevent
a protocol from terminating without having to explore the space of all executions.

The time complexity of computing a non-terminating, and thus infinite,
execution is necessarily infinite. We will instead look at the more
tractable problem of computing an extension for some fixed number of
steps $s$, which is the justification for defining obfuscated
consensus as satisfying termination and validity always but possibly
failing agreement.
As described in Section~\ref{section-model}, 
we replace an abstract adversary function with
a concrete adversary program that takes as input the starting configuration of
the system and the programs used by the processes to choose their next
operations and, ultimately, their outputs.

Formally, let us define the following problem. We assume that each
program $P_i$ has two special return actions $R_0$ and $R_1$, that it
must choose between after exactly $s$ steps; these provide the output
of each process in the consensus protocol. The adversary, a program
that takes as input the programs $P_1,\dots,P_n$ and the step bound
$s$ and constructs a schedule in which each process takes exactly $s$
steps, wins if any of the processes' outputs disagree.

The following theorem shows that the cost to the adversary to find a
bad execution, relative to the worst-case cost of each process to
carry out some execution, scales linearly with the number of steps
performed by each process.

\begin{theorem}
    \label{theorem-adversary-upper-bound}
    Given inputs
    $P_1,\dots,P_n$ representing processes in a shared-memory system 
    where, in all executions, (a) each process $p_i$ outputs a decision value in
    $s$ steps, and (b) each process $p_i$ uses a total of at most $t$ time to
    compute its transitions, there is an adversary program that
    computes in $O(st)$ time a schedule that causes some pair of processes to output different decision
    values.
\end{theorem}

To prove the theorem, we consider a restricted version of the FLP bivalence 
argument in which valence is
defined only in terms of solo extensions of the current configuration.
This will significantly reduce the number of extensions the adversary needs to
consider.

\newFunc{\pref}{pref}
\newFunc{\decision}{decision}

For each configuration $C$, define the \concept{preference}
$\pref_p(C)$ of $p$ in $C$ to be the output value of $p$ in the execution
$Cα$ where only $p$ takes steps before it decides. We will call such
an extension $α$ $p$'s solo-terminating extension and write
$\pref_p(C) = \decision_p(Cα)$ to indicate that $p$ decides this value
in configuration $Cα$.
We know that a unique such extension $α$ exists
because $p$ is deterministic (making $α$ unique) and
the system is wait-free (so that no fairness requirement prevents
$p$ from running alone).

Following Gelashvili~\cite{Gelashvili2018}, call a configuration $C$
\concept{solo-bivalent} if there are processes $p$ and $q$ such that
$\pref_p(C) ≠ \pref_q(C)$.
Call a configuration \concept{solo-$b$-valent} if $\pref_p(C) = b$
for all $p$.
Similar to the original FLP construction, our goal is to
start in a solo-bivalent configuration and stay in a solo-bivalent
configuration.

We start with some simple observations about how a process's
preference can change. Given executions $α$ and $β$, write $α \sim_p
β$ ($α$ is \concept{indistinguishable} by $p$ from $β$) if $α|p =
β|p$, meaning that $p$ observes the same events in both executions.
Note that $α \sim_p β$ in particular implies $\decision_p(α) =
\decision_p(β)$.

\begin{lemma}
    \label{lemma-persistence-of-preference}
    Let $C$ be a configuration and let $π$ be an operation of $p$.
    Then $\pref_p(Cπ) = \pref_p(C)$.
\end{lemma}

\begin{proof}
    Let $α$ be $p$'s solo-terminating extension of $Cπ$.
    Then $πα$ is a $p$'s solo-terminating extension of $C$, 
    and $\pref_p(C) = \decision_p(Cπα) =
    \pref_p(Cπ)$.
\end{proof}

\begin{lemma}
    \label{lemma-preference-read}
    Let $π$ be a read operation by $p$. Then $\pref_q(Cπ) =
    \pref_q(C)$.
\end{lemma}

\begin{proof}
    Let $α$ be $q$'s solo-terminating extension of $C$. Then $Cπα \sim_q Cα$
    implies $\pref_q(Cπ) = \decision_q(Cπα) = \decision_q(Cα) =
    \pref_q(C)$.
\end{proof}

\begin{lemma}
    \label{lemma-preference-write-same-location}
    Let $π_p$ and $π_q$ be writes by $p$ and $q$ to the same register.
    Then $\pref_p(Cπ_qπ_p) = \pref_p(Cπ_p)$.
\end{lemma}

\begin{proof}
    Since the value written by $π_q$ is replaced by $π_p$, we have
    $Cπ_q π_p \sim_p Cπ_p$ and thus
    $\pref_p(Cπ_q π_p) = \pref_p(C π_p) = \pref_p(C)$.
\end{proof}

\begin{lemma}
    \label{lemma-bivalent-initial-configuration}
    For any wait-free shared-memory 
    protocol satisfying termination and validity, there exists
    an initial solo-bivalent configuration for any $n≥2$.
\end{lemma}

\begin{proof}
    Take any configuration $C$ where two processes $p$ and $q$ have
    different inputs. If either process runs alone, it observes only
    its own input and is forced to decide it by validity.
\end{proof}

\begin{lemma}
    \label{lemma-preserve-bivalence}
    Let $C$ be a solo-bivalent configuration
    where $p$ and $q$ are processes with $\pref_p(C) ≠ \pref_q(C)$.
    Let $π_p$ and $π_q$ be the pending operations of $p$ and $q$ in
    $C$. Then, at least one of the following holds:
    \begin{enumerate}
        \item $Cπ_p$ is solo-bivalent with $\pref_p(Cπ_p) =
            \pref_p(C)$ 
            \\ and $\pref_q(Cπ_p) = \pref_q(C)$.
            \label{solo-bivalent-case-p}
        \item $Cπ_q$ is solo-bivalent with $\pref_p(Cπ_q) =
            \pref_p(C)$ 
            \\ and $\pref_q(Cπ_q) = \pref_q(C)$.
            \label{solo-bivalent-case-q}
        \item $Cπ_p π_q$ is solo-bivalent with $\pref_p(Cπ_p π_q) ≠
            \pref_p(C)$ 
            \\ and $\pref_q(Cπ_p π_q) ≠ \pref_q(C)$.
            \label{solo-bivalent-case-pq}
    \end{enumerate}
\end{lemma}

\begin{proof}
    By Lemma~\ref{lemma-preference-read}, $\pref_p(Cπ_p) = \pref(C)$, so
    to avoid case \iref{solo-bivalent-case-pq}, we must have
    $\pref_q(Cπ_p) ≠ \pref_q(C)$.

    Similarly, if case \iref{solo-bivalent-case-q} does not hold, we
    must have $\pref_p(Cπ_q) ≠ \pref_q(C)$.

    If neither of these cases holds, then we have both
    \begin{align*}
        \pref_p(Cπ_q) &≠ \pref_p(C)
        \intertext{and}
        \pref_q(Cπ_p) &≠ \pref_q(C).
    \end{align*}
    From Lemma~\ref{lemma-preference-read}, neither $π_p$ nor $π_q$ is
    a read. From Lemma~\ref{lemma-preference-write-same-location},
    $π_p$ and $π_q$ cannot be writes to the same location. So $π_p$
    and $π_q$ are writes to different locations. But then 
    $C π_p π_q$ and $C π_q π_p$ yield the same configuration. It follows
    that
    \begin{align*}
        \pref_p(C π_p π_q) &= \pref_p(C π_q π_p ) = \pref_p(C π_q) ≠
        \pref_p(C)
        \intertext{while}
        \pref_q(C π_p π_q ) &= \pref_q(C π_p) ≠ \pref_q(C).
    \end{align*}
    So case \iref{solo-bivalent-case-pq} holds.
\end{proof}

One difference between solo-bivalence and bivalence that is
illustrated by the third case of Lemma~\ref{lemma-preserve-bivalence}
is that it is not necessarily the case that a configuration with a solo-bivalent successor is itself solo-bivalent. But we are happy as
long as we can find solo-bivalent extensions of solo-bivalent
configurations, even if this means passing through intermediate
configurations that may not be solo-bivalent.

To turn Lemma~\ref{lemma-preserve-bivalence} into an algorithm, we
just need to be able to test which of its three cases holds. To
compute a bad schedule, the adversary carries out the following steps:
\begin{enumerate}
    \item Pick two processes $p$ and $q$, and start in a configuration $C$ where $\pref_p(C) = 0$ and
        $\pref_q(C) = 1$, as in Lemma~\ref{lemma-bivalent-initial-configuration}.
    \item Given $C$ with $\pref_p(C) ≠ \pref_q(C)$:
        \begin{enumerate}
            \item If $\pref_q(Cπ_p) = \pref_q(C)$, append $π_p$ to the
                schedule and set $C←Cπ_p$.
            \item If $\pref_p(Cπ_q) = \pref_p(C)$, append $π_q$ to the
                schedule and set $C←Cπ_q$.
            \item Otherwise, the third case of
                Lemma~\ref{lemma-preserve-bivalence} holds.
                Append $π_pπ_q$ to the schedule and set $C←Cπ_pπ_q$.
        \end{enumerate}
    \item Repeat until one of $p$ and $q$ decides; then run the other
        to completion, appending its operations to the schedule.
\end{enumerate}

The cost of computing $\pref_p(C)$ for any configuration
$C$ is $O(t)$, since we can just simulate $p$ until it decides. To
avoid copying, we do this simulation in place, logging any changes to
$p$'s state or the state of the shared memory so that we can undo them
in $O(t)$ time. For computing $\pref_p(Cπ_q)$, we again
pay at most $O(t)$ time, since applying $π_q$ to $C$ takes no more than
$O(t)$ time, as does running $p$ to completion thereafter. The same
bounds hold with the roles of $p$ and $q$ reversed.

It follows that each iteration of the main body of the algorithm
finishes in $O(t)$ time. Since each of $p$ and $q$ do at most $s$
steps before deciding, there are at most $2s = O(s)$ iterations,
giving a total cost of $O(st)$, even taking into account the $O(t)$ cost
of the last phase of the algorithm. Since every iteration of the
algorithm yields a solo-bivalent configuration, the preferences of $p$
and $q$ are never equal, and so the constructed schedule
violates agreement.
This concludes the proof of
Theorem~\ref{theorem-adversary-upper-bound}.

\section{Obfuscated consensus from obfuscated threshold}
\label{section-consensus-from-threshold}

The ratio between the $O(t)$ time complexity of each process in
Theorem~\ref{theorem-adversary-upper-bound} and the $O(st)$ time
complexity of the adversary suggests that a useful strategy for the
processes faced with a computationally-limited adversary is to exhaust
the adversary's resources by forcing it to simulate a large number of
expensive potential executions. We can do this by making the outcome
of each process's execution depend on evaluating an expensive function of the
view that process has in the execution.

An \concept{obfuscated threshold function} $f:\Set{0,\dots,\ell} →
\Set{0,1}$ is a function, represented as a random-access machine
program, with $f(0) = 0$ and $f(\ell) = 1$. The \concept{obfuscated
threshold problem} is to find, given a representation of such a function,
an input $v$ such that $f(v) ≠ f(v+1)$. We will show that finding a
bad execution of a particular length of any obfuscated consensus
protocol is equivalent to solving this problem.  

Binary search will find a $v$ with $f(v)≠f(v+1)$ in $O(\log \ell)$
evaluations of $f$, which puts an upper bound on the cost of solving
the obfuscated threshold function problem. But it may be that a
sufficiently clever algorithm can extract $v$ from $f$ without
evaluating $f$ explicitly. We will discuss how to build an $f$ that
resists such attacks in
Section~\ref{section-obfuscation-implementation}.

Note that the end conditions $f(0) = 0$ and $f(\ell)=1$ show that at
least one transitional value $v$ always exists, and, for some
choices of $f$, many such transitional values exist. This makes the
name ``obfuscated threshold function'' a bit misleading, since $f$ is
not necessarily a threshold function. But choosing a threshold
function in particular reduces the likelihood of guessing a bad $v$ at
random, so we will try to produce threshold functions if we can.

We will show a reduction between these problems
by constructing an obfuscated consensus protocol that
uses an obfuscated threshold function to compute its output values. In
Algorithm~\ref{alg-consensus-from-threshold}, we give a protocol that
assumes that every process (and thus also the adversary) has access to
an obfuscated threshold function $f$, and show that the adversary
computes a bad execution for this protocol if and only if it can find
a $v$ such that $f(v) ≠ f(v+1)$.

\newFunc{\ObfuscatedConsensus}{obfuscatedConsensus}

\begin{algorithm}
    \SharedData{
        array $A[1\dots s][2]$ of atomic registers, initialized to $⊥$\;
    }
    \Procedure{$\FuncSty{approximateAgreement}(\Input,s)$}{
        $i ← \Input$\;
        \For{$r ← 1\dots s$}{
            \tcp{write my current position}
            $A[r][i \bmod 2] ← i$\;
            \tcp{read opposite position}
            $i' ← A[r][(i+1) \bmod 2]$\;
            \eIf{$i'=⊥$}{
                \tcp{no opposition, stay put}
                $i ← 2⋅i$\;
            }{
                \tcp{adopt midpoint}
                $i ← i + i'$\;
            }
        }
        \Return $i$\;
    }
    \Procedure{$\ObfuscatedConsensus(\Input,f,s)$}{
        $i←\FuncSty{approximateAgreement}(\Input,s)$\;
        \Return $f(i)$\;
    }
    \caption{Obfuscated consensus based on an obfuscated threshold}
    \label{alg-consensus-from-threshold}
\end{algorithm}

Pseudocode for a protocol that does this is given in
Algorithm~\ref{alg-consensus-from-threshold}. The main procedure
$\ObfuscatedConsensus$ takes as arguments a consensus input value
$\Input∈\Set{0,1}$, a representation of a function
$f$ that converts integer values to Boolean decisions,
and a security parameter $s$ that scales the number of
shared-memory steps taken by the process. 

The protocol proceeds in two phases.

First, the processes carry out
a $s$-round approximate agreement protocol that assigns each process
$p$ a value $v_p$ in
the range $0 \dots 2^s$, with the properties that
\begin{enumerate}
    \item Any process that sees only input $v$ obtains value $2^s⋅v$.
    \item For any two processes $p$ and $q$, $\abs*{v_p - v_q} ≤ 1$.
\end{enumerate}

The approximate agreement protocol is an adaptation of Moran's
one-dimensional midpoint protocol~\cite{Moran1995}.
In Moran's original protocol, each process takes a sequence of
snapshots of current values and adopts the midpoint of the values it sees until all fit within a particular range. In our protocol, we
organize this sequence into a layered execution of $s$ rounds, where
in each round, a process adopts the average value it sees, scaled by a
factor of $2$ per round to track the integer numerators instead of the
actual fractional values. By taking
advantage of the inputs starting at $0$ and $1$, we can show that
at most two values, differing by at most $1$, appear at the end of
each round; this allows us to replace the snapshot with a pair of
registers, one of which holds the even value and one the odd.
Formally:

\begin{lemma}
    \label{lemma-approximate-agreement}
    Let $S_r$ be the set of all $i$ values that are held by any process after $r$
    iterations of the loop in procedure $\FuncSty{approximateAgreement}$.
    Then $S_r ⊆ \Set{v_r, v_r+1}$ for some $0≤v_r≤2^r-1$.
\end{lemma}

\begin{proof}
    By induction on $r$. For $r=0$, we have $v_0=0$ and $S_0 ⊆ \Set{0,1}$.

    For $r+1$, from the induction hypothesis, there is a value $v_r$
    such that $S_r ⊆ \Set{v_r,v_r+1}$. So, each process writes either
    $v_r$ or $v_{r+1}$ to one of the registers $A[r][-]$ before
    reading the other register; since $v_r$ and $v_r+1$ have
    different values mod $2$, they will not overwrite each other.

    It follows that whichever value $v_r$ or $v_{r+1}$ is
    written first will be visible to all processes. If this is
    $v_r$, then every process either sees $v_r$ alone and chooses a
    new value $2⋅v_r$, or sees $v_r$ and
    $v_{r+1}$ and chooses $v_r + (v_r+1) = 2⋅v_r + 1$. In this case,
    we get $S_{r+1} ⊆ \Set{v_{r+1}, v_{r+1} + 1}$ where $v_{r+1} =
    2⋅v_r$. The case where $v_r+1$ is written first is similar, except
    now the processes all choose either $v_{r+1} = 2⋅v_r + 1$ or
    $v_{r+1}+1 = 2⋅(v_r+1)$.
\end{proof}

In particular, the set of output values is given by $S_s ⊆
\Set{v_s,v_s+1}$ where $0≤v_s≤2^s-1$.

To obtain its decision value, each process $p$ then feeds its output
$v_p$ to a function $f(v_p)$. To ensure validity, we
require that $f(0) = 0$ and $f(2^s) = 1$. Agreement is
obtained if $f(v_p) = f(v_q)$ for all processes $p$ and $q$.

If $f$ requires $T$ time to evaluate, the time complexity incurred
by each process running Algorithm~\ref{alg-consensus-from-threshold}
is $O(s)+T$. From Theorem~\ref{theorem-adversary-upper-bound} this
implies that there is an adversary strategy that computes a bad
execution in time $O(s^2 + sT)$; this also tells us that at least one
bad execution exists. We'd like to show a comparable lower bound on
the cost to the adversary, assuming it is difficult to find a
threshold in $f$.

Since every process has a value in $\Set{v_s,v_s+1}$, agreement holds
automatically if $f(v_s) = f(v_s+1)$. So an adversary that finds an
execution that produces disagreement also finds a value $v_s$ with
$f(v_s) ≠ f(v_s+1)$.

Thus, restricting the adversary and processes to have only oracle access to $f$ induces the cost separation, where each process computes $f$ only once, and the adversary has to search for the threshold.

For any fixed $f$ with $f(0) = 0$ and $f(2^s) = 1$, there is a
random-access machine that outputs a threshold value $v$ with $f(v) ≠ f(v+1)$ in $O(1)$ time, since it can just include $v$ in its code. We
will need to avoid this by considering the average-case cost to find a
threshold value for an $f$ drawn from some family of threshold
functions. Applying Algorithm~\ref{alg-consensus-from-threshold} to
the functions in this family then give a family of consensus
protocols.

\begin{theorem}
    \label{theorem-consensus-from-threshold}
    Let $F$ be a family of threshold functions with range
    $\Set{0,\dots,\ell}$. Fix some $n$ and some $s$ 
    such that $2^s ≥ \ell$, and for each
    $f∈F$, let $P_f$ be the obfuscated consensus protocol given by
    instantiating the $s$-round version of
    Algorithm~\ref{alg-consensus-from-threshold} with $f$ as its
    decision function.

    Suppose we sample $P_f$ uniformly at random and consider
    a starting configuration in which not all processes have the same input.
    If there is a
    probabilistic random-access machine $M$ that takes $P_f$ and the
    process's initial states as input,
    runs for $T$ expected time, and outputs a schedule that produces
    disagreement, then there is a
    probabilistic random access machine $M'$ that takes $f$ as input,
    runs for $T+O(ns)$ expected time, and outputs a $v$ with $f(v) ≠
    f(v+1)$.
\end{theorem}

\begin{proof}
    We give a construction for $M'$.

    Run $M$ on $P_f$ to obtain a schedule for
    Algorithm~\ref{alg-consensus-from-threshold}, using $T$ time. Simulate the
    algorithm until each process $p$ returns a value $i_p$ from
    $\FuncSty{approximateAgreement}$, but do not have any process
    call $f$. this takes $O(ns)$ time since each simulated step can be
    done in constant time in the RAM model. Output the smallest $i_p$
    value as $v$.

    The cost of this procedure is bounded by $T+O(ns)$. From
    Lemma~\ref{lemma-approximate-agreement} we know that all processes
    obtain either $v$ or $v+1$, and in the case that $M$ outputs a
    schedule that produces disagreement, we have $f(v) ≠ f(v+1)$.
\end{proof}

We clarify that the sampling is used only to state average-case hardness over protocol instances. After $f$ is sampled, the protocol $P_f$ is deterministic. This prevents the adversary from knowing the hidden threshold a head of time from a previous execution of the protocol.

In the following section, we give an implementation of a family
of threshold functions with these properties.

\section{Implementing an obfuscated threshold function}
\label{section-obfuscation}

Because the adversary can observe the processes' programs, we cannot enforce oracle-only access to $f$. 
We therefore need an obfuscated implementation of $f$ that is expensive to evaluate and reveals, 
when evaluated at $v$, no information about $f(v')$ beyond what follows from the threshold-function property.

\subsection{Definition}
\label{section-obfuscation-construction}

We introduce a new approach for computationally obfuscating programs that use one-way function primitives. 
The program's creation uses randomization that can be revealed by (tunable) inversion processing required to invert the one-way function. 
We construct the obfuscation by using a random oracle to hide a sequence of bits that encode the truth table of a threshold function. 
To obtain the output from the obfuscation, one must find the pre-image of a single hash. 
Thus it is imperative that the hardness of computing the pre-image be tunable such that it is feasible to find a single pre-image, yet infeasible for the computationally bounded adversary to check enough pre-images to de-obfuscate the program. 
We achieve this tunability not by weakening the hash function, but by providing a prefix for the pre-image, such that the processes need only complete it. We begin with the following definitions.

\begin{definition}[Threshold Function]
    \label{definition-threshold-function}
    Let $T,\ell\in\mathbb{N}$ be two integers such that $0 < T < \ell $. Define function $f_T:[\ell] \rightarrow \{0,1\}$ as 
    $$f_T(i) = \begin{cases}0 & i < T \\1 &  i\geq T\end{cases}$$
    Then $f_T$ is a threshold function for threshold $T$. 
\end{definition}

\begin{definition}[Random Oracle Hash Function]
    \label{definition-hash-function}
    Let $m\in\mathbb{N}$ and let $H_m:\{0,1\}^m \rightarrow \{0,1\}^m$ be sampled uniformly at random from the set of all functions $\{0,1\}^m \to \{0,1\}^m$ (i.e., $H_m$ is a \emph{random oracle} / random function).
    Then for all distinct $x\neq x'\in\{0,1\}^m$ and all $y,y'\in\{0,1\}^m$,
    \begin{enumerate}
        \item $\Pr \left[H_m(x)=y\right] = 2^{-m}$
        \item $\Pr \left[H_m(x)=y \ \wedge\ H_m(x')=y'\right] = 2^{-2m}$
    \end{enumerate}
\end{definition}

\Cref{definition-hash-function} is adapted from Katz and Lindell~\cite{KatzLindell}, Chapter~5.1. Brute force search needs at most $2^m$ trials to find a pre-image; the next corollary states that a known pre-image prefix reduces this search space.

\begin{corollary}[Tunable pre-image runtime]
    \label{corollary-tunable-hardness}
    Let $H_m:\{0,1\}^m \to \{0,1\}^m$ be a random oracle (\Cref{definition-hash-function}).
    Fix $k<m$ and a prefix $p\in\{0,1\}^k$, and let $H_m(p\Vert r) = y$ be the output of $H_m$ for some extension $r$ of $p$. For any (possibly randomized) algorithm $\mathcal{X}$ that makes at most $q$ oracle queries to $H_m$ and outputs $z\in\{0,1\}^{m-k}$, it holds that
    $$ \Pr_{H_m,\mathcal{X}}[H_m(p\Vert z)=y] \leq \frac{q}{2^{m-k}} $$
\end{corollary}

The parameter $k$ here is used to tune the size of the pre-image search space, by exposing the first $k$ bits of the pre-image later in the obfuscation construction.

\begin{definition}[Threshold obfuscation scheme]
    \label{definition-obfuscation-scheme}
    Let $\Pi=(Prep,Probe)$ be a pair of algorithms.
    Given a representation of a threshold function $f_T$ (\ref{definition-threshold-function}), if $Prep$ produces output $O_T = Prep(f_T)$ such that, for any given input $x$ of $f_T$, $Probe(O_T,x)=f_T(x)$ holds with high probability, then $\Pi$ is a threshold obfuscation scheme.
\end{definition}

We move one to our security definition. We base this definition on the adversarial experiment format used in the book by Katz \& Lindell \cite{KatzLindell}. There exists a definition for obfuscation in  \cite{DBLP:journals/jacm/BarakGIRSVY12} (Boaz Barak et al.). This definition requires that no probabilistic polynomial-time adversary can generate the same outputs as the original program with high probability. While we would like to have such obfuscation for the threshold function, our current obfuscation construction does not satisfy this very strong notion; rather, we require that to reproduce the same outputs, an adversary must have a runtime that is at least logarithmic in the size of the construction, and an adversary must reproduce the exact outputs with probability $1$; otherwise, it cannot do any better than a random guess. Specifically, since we tailor this obfuscation to a threshold function, we formulate this specific notion of obfuscation as an experiment where the adversary gets the obfuscated threshold function just like the processes, and must recover the threshold $T$ in order to succeed.

\begin{definition}[Obfuscation Security Experiment $Obf_{\mathcal{A},\Pi,\ell}$]
    \label{definition-obfuscation-experiment}
    Let $\ell\in\mathbb{N}$ be the range of values for threshold $T$.
    Let $\Pi=(Prep,Probe)$ be an obfuscation scheme (\ref{definition-obfuscation-scheme}).
    Let $\mathcal{T}_{Probe}$ be the runtime of algorithm $Probe$.
    Let $\mathcal{A}$ be an adversarial algorithm with runtime $\mathcal{T_A} < \mathcal{T}_{Probe} \cdot \log \ell$.

    \begin{enumerate}
        \item A random $T\in[\ell]$ is chosen.
        \item Compute $O_T = Prep(f_T)$.
        \item $\mathcal{A}$ is given $O_T$.
        \item $\mathcal{A}$ returns an integer $T^\prime$.
        \item Let $\tau$ be the number of times $\mathcal{A}$ computes $H_m$.
        \item The output of the experiment is $1,\tau$, if $T^\prime = T$, otherwise the output is $0,\tau$.
    \end{enumerate}
\end{definition}

Using \Cref{definition-obfuscation-experiment}, we formulate a definition of security for a threshold obfuscation scheme. Defining the security as an upper bound on the probability of success in the experiment $Obf$ would yield a strong definition. However, in our case, since the adversary must be able to cause the protocol to terminate with certainty, we can relax the definition and require that the expected number of hashes the adversary must compute in order to find $T$ is bounded. This approach allows us to simplify the security proof.

\begin{definition}
    \label{definition-obfuscation-secure}
    A threshold obfuscation scheme (\Cref{definition-obfuscation-scheme}) $\Pi=(Prep,Probe)$ is secure if for every adversary algorithm $\mathcal{A}$, it holds that
    $$ \mathbb{E} \left[ \tau \mid b = 1\right] \geq  \frac{1}{2} \left( (2^{m-k}+1) \log(\ell+1) \right)$$
    where $b,\tau = Obf_{\mathcal{A},\Pi,\ell}$.
\end{definition}

Note that increasing $k$ exposes a greater portion of the pre-image, shrinking the search space, and reducing the expected number of probes required to find $T$. While a bound on the probability of success would be stronger, this definition is sufficient to show that the adversary must do more work than the processes in order to cause the protocol to fail. We move on to define our threshold obfuscation scheme. We begin with the $Prep$ algorithm, the output of which will be used to compute the output of the threshold function in the $Probe$ algorithm. We call this algorithm preprocess.\\

\begin{AlgDescBox}{preprocess}
    \begin{enumerate}
        \item Fix integers $1 < k < m$ and $\ell \geq 2$.
        \item Choose threshold $1 < T  < \ell$ uniformly at random.
        \item For all $1 \leq i \leq \ell$, sample binary strings $P_i\leftarrow\{0,1\}^k$ and $r_i\leftarrow\{0,1\}^{m-k}$ uniformly at random.
        \item For all $1 \leq i \leq \ell$
            \begin{enumerate}
                \item compute $d = H(P_i\Vert r_i)$
                \item set $C_i \gets d_1 d_2 \dots d_{m-1}$ that is, $C_i$ is the first $m-1$ bits of $d$
                \item set $v_i \gets d_m \oplus f_T(i)$
            \end{enumerate}
        \item Return $(P_1,\dots,P_\ell, C_1,\dots,C_\ell,v_1,\dots,v_\ell)$.
    \end{enumerate}
\end{AlgDescBox}

\noindent
Next, we define the $Probe$ algorithm, which we will use to compute a single output of $f_T$. \\

\begin{AlgDescBox}{probe}
    \begin{enumerate}
        \item Assume access to $H_m$ as well as $C_1,\dots,C_\ell$, $P_1,\dots,P_\ell$, and $v_1,\dots,v_\ell$. 
        \item Let $1\leq i \leq \ell$ be an input of $f_T$.
        \item Find string $r \in \{0,1\}^{m-k}$ such that $[H_m(P_i\Vert r)]_j = [C_i]_j$ for all $j=1,\dots,m-1$.
        \item Return $[H_m(P_i\Vert r)]_m \oplus v_i$
    \end{enumerate}
\end{AlgDescBox}

\noindent
A more complete definition and implementation can be found in \Cref{section-obfuscation-implementation}.

\subsection{Correctness $\&$ security}
\label{section-obfuscation-security-correct}

We show that Algorithms~\ref{alg-preprocessing} and~\ref{alg-probeThreshold}, respectively $\FuncSty{preprocess}$ and $\FuncSty{probe}$, form a threshold obfuscation scheme (\ref{definition-obfuscation-scheme}) satisfying Definition~\ref{definition-obfuscation-secure}: \Cref{claim-is-obf} proves correctness and \Cref{claim-secure} proves security. We begin with a lemma.

\begin{lemma}[Distinct pre-image with prefix]
    \label{lemma-distinct-pre-image}
    Let $k,m\in\mathbb{N}$ be two integers where $k<m$, and let $H_m$ be a random oracle \ref{definition-hash-function}.
    Let $p\in\{ 0,1\}^k, \quad r\in \{0,1\}^{m-k}$ be two binary strings. It holds that
    $$\Pr\left[ \exists_{x\neq r}( H_m(p||x) = H_m(p||r) ) \right] \leq 2^{-k} $$
\end{lemma}

\begin{proof}
    Let $p\in\{ 0,1\}^k, \quad r\in \{0,1\}^{m-k}$ be two binary strings.
    Choose uniformly at random $x\in \{0,1\}^{m-k} \setminus \{r\}$.
    By \Cref{definition-hash-function}, it holds that
    $$ \Pr[H_m(p\Vert r) = H_m(p \Vert x)] \leq 2^{-m} $$
    \\ \\ \\
    By the union bound
    $$ \Pr[\exists x(H_m(p\Vert r) = H_m(p \Vert x))] = \Pr [ \bigcup_{x \neq r}H_m(p\Vert r) = H_m(p \Vert x)]$$ 
    $$ \leq \sum_{x \neq r}\Pr [H_m(p\Vert r) = H_m(p \Vert x)] $$
    Since there are $2^{m-k} - 1$ possible choices for $x$ and assuming each sample of $x$ happens independently of the others,
    $$ = (2^{m-k}-1)\cdot \Pr [H_m(p\Vert r) = H_m(p \Vert x)] = (2^{m-k}-1) \cdot 2^{-m} \leq 2^{-k}  $$
    Hence, we obtain that 
    $$ \Pr[\exists x(H_m(p\Vert r) = H_m(p \Vert x))] \leq 2^{-k} $$
\end{proof}

The following is the construction correctness claim. We require that the parameters $k,\ell$ satisfy $\frac{\ell}{2^k} = negl(k)$, meaning that their ratio is negligible in $k$. This assumption is reasonable since increasing $k$ decreases the probability of a collision and reduces the amount of work required to find the suffix of a pre-image.

\begin{claim}{isObf}
    \label{claim-is-obf}
    The algorithm pair $(\FuncSty{preprocess},\FuncSty{probe})$ is a threshold obfuscation scheme \ref{definition-obfuscation-scheme}, for $\ell,k$ that satisfy $\frac{\ell}{2^k} = negl(k)$.
\end{claim}

\begin{proof}
    Let $\ell \in \mathbb{N}$ be some positive integer and let $f_T$ be some threshold function where $0 < T < \ell$.
    Fix some $m,k$ such that $0<k<m$.
    Let $C,P = \FuncSty{preprocess}(\ell,T,m,k)$ be the outputs of algorithm $\FuncSty{preprocess}$.
    For each possible input $i\in[\ell]$ of $f_T$, the algorithm executes the following steps:\\
    In lines \ref{line-preprocessing-assign-P}, \ref{line-preprocessing-assign-r}, the algorithm picks random strings $P[i] \in \{0,1\}^k$ and $r_i \in \{0,1\}^{m-k}$. Their concatenation $P[i]\Vert r_i$ is a string of length $m$ which is a valid input for $H_m$.\\
    In \ref{line-preprocessing-assign-D}, the algorithm computes the hash of $d \gets H_m(P[i]\Vert r_i)$.\\
    In \ref{line-preprocessing-assign-C}, the algorithm stores the first $m-1$ bits of $d$ in $C[i]$.\\
    In \ref{line-preprocessing-assign-V}, the algorithm stores $f_T(i)$ in $v_i$ by assigning $v_i \gets d_m \oplus f_T(i)$.
    
    Therefore, for any input $i$ of $f_T$, there exists at least one string $r_i\in\{0,1\}^{m-k}$ such that $H_m(P[i]\Vert r_i)[j]=C[i][j]$ for all $j=1,\dots,m-1$ and $H_m(P[i]\Vert r_i)[m] \oplus f_T(i) = V[i]$.\\

    Algorithm $\FuncSty{probe}$ has access to $\ell,m,k,C,P,V$ and is given input $i$ of the threshold function.
    The algorithm finds suffix $y$ such that $H_m(P[i]\Vert y)[1,\dots, m-1] = C[i][1,\dots,m-1]$.
    It remains to show that $H_m(P[i]\Vert y)[m] \oplus V[i] = f_T(i)$ holds with high probability.
    Assume the contrary.\\
    This implies $H_m(P[i]\Vert y)[m] \oplus H_m(P[i]\Vert r_i)[m] \oplus f_T(i) \neq f_T(i)$ which can only happen if $H_m(P[i] \Vert y)[m] \neq H_m(P[i] \Vert r_i)[m]$.
    Since $H_m$ is deterministic, this implies that $y \neq r_i$.
    But by \Cref{lemma-distinct-pre-image}, the probability that this occurs is at most $2^{-k}$. Thus, for a large enough $k$, it holds that $y = r_i$ with high probability.
    Thus, $H_m(P[i] \Vert y) = H_m(P[i]\Vert r_i)[m]$, which means that, w.h.p
    $$H_m(P[i] \Vert y) \oplus V[i] = H_m(P[i]\Vert r_i)[m] \oplus H_m(P[i]\Vert r_i)[m] \oplus f_T(i) = f_T(i)$$
\end{proof}

\begin{lemma}
    \label{lemma-expected-probes}
    For any adversary $\mathcal{A}$ that recovers $T$ from $C,P,V$, it holds that
    $$\mathbb{E}[\tau] \geq \frac{1}{2}\left( (2^{m-k}+1) \cdot \log(\ell+1) \right)$$
    where $\tau$ is the number of times $\mathcal{A}$ computes $H_m$.
\end{lemma}

\begin{proof}
    By \Cref{definition-hash-function}, given $C,P,V$, the threshold $T$ that is hidden in $C,P,V$ can be any value in $\{1,\dots,\ell\}$ with equal probability.
    Moreover, \Cref{definition-hash-function} implies that for each $i\in[\ell]$, every suffix $y\in\{0,1\}^{m-k}$ has equal probability of satisfying $H_m(P[i]\Vert y)[1,\dots,m-1]=C[i]$.
    It follows that any adversarial algorithm that attempts to compute $T$ directly from $C,P,V$ must check $H_m(P[i] \Vert y)[1,\dots,m-1] = C[i]$ for some number of strings $y$.
    Order the set $\{0,1\}^{m-k}$ by some arbitrary order (by the number each string represents for example).
    Denote $X_{i,j} = H_m(P[i] \Vert y)$ where $y$ is the $j^{th}$ string in $\{0,1\}^{m-k}$. 
    Denote $r=2^{m-k}$.
    Let $M\in\{0,1, \bot,?\}^{\ell \times r}$ be the matrix such that for all $i,j$:
    \begin{itemize}
        \item $M_{i,j} = 1 \iff$ $\mathcal{A}$ discovers that $X_{i,j} [1,\dots,m-1] = C[i] \wedge X_{i,j}[m] \oplus V[i] = 1$
        
        \item $M_{i,j} = 0 \iff$ $\mathcal{A}$ discovers that $X_{i,j} [1,\dots,m-1] = C[i] \wedge X_{i,j}[m] \oplus V[i] = 0$
 
        \item $M_{i,j} = \bot \iff$ $\mathcal{A}$ discovers that $X_{i,j} [1,\dots,m-1] \neq C[i]$
        
        \item $M_{i,j} = ? \iff$ $\mathcal{A}$ does not know if the $j^{th}$ string satisfies the condition for $C[i],P[i]$
    \end{itemize}
    
    Notice that:
    \begin{itemize}
        \item [(a)] Since the adversary only needs to find one suffix $y$ per column $i$ that satisfies $X_{i,j} [1,\dots,m-1] = C[i]$, it holds that $M_{i,j} = 1 \vee 0 \implies \forall_{g\neq j} M_{i,g} = \bot$.
        \item [(b)] Since $C,P,V$ encode the truth table of a threshold function, the adversary can know that if $M_{i,j} = 0$ then for all $1\leq r \leq i$ it holds that $M_{r,j'} = 0$ for some $j'$, meaning $f_T(r)=0$. 
    \end{itemize}
    We call the process of checking if a single suffix $y$ satisfies $X_{i,j}[1,\dots,m-1]=C[i]$ a probe of $M$ (not to be confused with the $probe$ algorithm).
    Denote $M^{(t)}$ as the state of $M$ after the adversarial algorithm makes $t$ probes. Conversely, denote $M^{(0)}$ as the initial state of $M$ at the start of the execution, this implies $M^{(0)}_{i,j} = ?$ for all $i,j$. 
    Let $L_t$ be the number of columns in $M^{(t)}$ with at least one $?$ cell (conversely, without any $1$ or $0$ cells). 
    Let $P_t$ be the number of non-$?$ cells, in columns with at least one $?$ in $M^{(t)}$.
    Define potential function $\Phi$ as follows
    $$\Phi(L,P) = \frac{1}{2}\left( (r+1) \cdot \log(L+1) - P \right)$$
    We show that $\Phi$ is the lower bound on the expected number of probes $\mathcal{A}$ must execute to find $T$.
    
    \begin{claim}
        \label{claim-halving-searchspace}
        Denote $\text{hit}_t$ as the event where the $t^{th}$ probe $\mathcal{A}$ makes is at index $i,j$ and it holds that $X_{i,j} [1,\dots,m-1] = C[i] \wedge (X_{i,j}[m] \oplus V[i] = 1 \vee X_{i,j}[m] \oplus V[i] = 0)$.
        Conditioned on $\text{hit}_t$, it holds that $\log(L_{t+1}+1) \geq \log(L_t + 1)-1$, that is, a successful probe reduces the number ofpossibilities for $T$ by at most half.
    \end{claim}

    \textit{Proof.} Since $\mathcal{A}$ knows that $x=\max \{ i \mid f_T(i)=0 \} \leq T \leq  \min \{ i \mid f_T(i)=1 \}$, it follows that $T$ is uniform on $\{x,x+1,\dots,x + L_t\}$, it follows that
    $$\Pr[f_T(x+j)=0]=\Pr[T>x+j]=\frac{j}{L_t + 1},$$
    $$ \Pr[f_T(x+j)=1]=\Pr[T\le x+j]=\frac{L_t + 1-j}{L_t + 1} $$
    If $f_T(x+j)=0$ then $T\in\{x,\dots,x+j-1\}$ and the new interval size is $j$.\\
    If $f_T(x+j)=1$ then $T\in\{x+j,\dots,x+L_t\}$ and the new interval size is $L_t + 1-j$.\\
    This implies that
    $$\mathbb{E}\left[\log(L_{t+1}+1) \mid \text{hit}_t\right]=\frac{j}{L_t + 1}\log j + \frac{L_t + 1-j}{L_t + 1}\log(L_t + 1-j)$$
    this expression is minizied at $j=\frac{L_t+1}{2}$, thus
    $$\mathbb{E}\left[\log(L_{t+1}+1) \mid \text{hit}_t\right] \geq \log((L_t + 1)/2) = \log(L_t + 1) - 1$$
    $$\mathbb{E}\left[\log(L_{t+1}+1) - \log(L_t + 1) \mid \text{hit}_t\right] \geq  -1 \Diamond$$

    Let $i$ be the index of the column where $\mathcal{A}$ makes the $t^{th}$ probe.
    Denote $q_t = |\{ j \mid M^{(t)}_{i,j} = \bot \}|$, that is, $q_t$ is the number of probes $\mathcal{A}$ on column $i$ that where {\bf not} a hit. 
    $$\Pr[\text{hit}_t] = \frac{1}{r-q_t}$$
    If $\text{hit}_t$ does not occur, then $P_{t+1} = P_t + 1$.\\
    If $\text{hit}_t$ does occur, then by (b), we have that $P_{t+1} \leq P_t + 1 - (q_t + 1) = P_t - q_t$.\\
    Denote $\Delta P = P_{t+1} - P_t$, it follows that
    $$ \mathbb{E}[\Delta P] \leq (1 - \Pr[\text{hit}_t]) + \Pr[\text{hit}_t]\cdot(-q_t) = 1 - \Pr[\text{hit}_t]\cdot(q_t + 1)$$
    $$ -\mathbb{E}[\Delta P] \geq \Pr[\text{hit}_t]\cdot(q_t + 1) - 1 \quad \text{(c)}$$ 
    Similarly, denote $\Delta \log (L + 1) = \log (L_{t+1} + 1) - \log (L_t + 1)$ and $\Delta\Phi = \Phi(L_{t+1},P_{t+1}) - \Phi(L_t,P_t)$. It follows that
    $$ \mathbb{E}[ \Delta\Phi ] = \frac{1}{2} \left( (r+1) \mathbb{E}[\Delta \log(L+1)] - \mathbb{E}[\Delta P] \right) $$
    $$ = \frac{1}{2} \left( (r+1) \Pr[\text{hit}_t] \cdot \mathbb{E}[\Delta \log(L+1) \mid \text{hit}_t] - \mathbb{E}[\Delta P] \right)$$
    By \Cref{claim-halving-searchspace} and (c),
    $$ \mathbb{E}[ \Delta\Phi ] \geq \frac{1}{2} \left( (r+1) \cdot \Pr[\text{hit}_t] \cdot (-1) - 1 + \Pr[\text{hit}_t] \cdot (q_t + 1) \right) $$
    $$= \frac{1}{2} \cdot \left( -1 + \Pr[\text{hit}_t](q_t - r) \right)$$
    $$= \frac{1}{2} \cdot \left( -1 + \frac{1}{r-q_t}(q_t - r) \right)$$
    Thus 
    $$\mathbb{E}[\Delta \Phi] \geq  \frac{1}{2}(-1 + (-1))=-1 \quad \text{(d)}$$

    Having bounded the expected change in $\Phi$, we bound the expected number of steps $t$ that $\mathcal{A}$ will need to make in order to find $T$.
    We begin by showing that $\Phi + t$ is submartingale.
    Notice that when $L_t = 0$ then the adversary knows all outputs of $f_T(\cdot)$, thus we define stopping time as 
    $$ \tau = \min \{ t \mid L_t = 0 \}$$
    By definition, if $L_\tau = 0$ then there are no unknown columns at time $\tau$; thus, $P_\tau = 0$, which implies that
    $$ \Phi(L_\tau,P_\tau) = \Phi(0,0) = \frac{1}{2} \left((+1) \log (0 + 1) - 0\right) = 0 $$
    Let $\mathcal{F}_t = \sigma(M^{(t)})$ be the filtration generated by $M$ after $t$ probes.
    Denote $\Phi_t = \Phi(L_t,P_t)$.
    By (d), it holds that 
    $$\mathbb{E}[\Phi_{t+1} -\Phi_t \mid \mathcal{F}_t] \geq -1$$
    $$\implies \mathbb{E}[\Phi_{t+1} -\Phi_t + 1 \mid \mathcal{F}_t] \geq 0$$
    $$\implies \mathbb{E}[ (\Phi_{t+1} + t) - (\Phi_t + t) + 1 \mid \mathcal{F}_t] \geq 0$$
    $$\implies \mathbb{E}[ (\Phi_{t+1} + (t+1)) - (\Phi_t + t)\mid \mathcal{F}_t] \geq 0$$
    Since $\Phi_t$ is completely determined by $M^{(t)}$, $\mathbb{E}[\Phi_t\mid \mathcal{F}_t] =  \Phi_t$.
    This implies that
    $$\mathbb{E}[ (\Phi_{t+1} + (t+1)) - (\Phi_t + t)\mid \mathcal{F}_t] =$$
    $$ \mathbb{E}[ (\Phi_{t+1} + (t+1)) \mid \mathcal{F}_t] - \mathbb{E}[(\Phi_t + t)\mid \mathcal{F}_t] \geq 0$$
    $$\implies \mathbb{E}[ (\Phi_{t+1} + (t+1)) \mid \mathcal{F}_t] \geq  \mathbb{E}[(\Phi_t + t)\mid \mathcal{F}_t] = \Phi_t + t \quad (e)$$
    Therefore, $\Phi_t + t$ is submartingale.
    Notice that for $\ell \cdot r$ probes, the adversary knows all columns of $M$ and the algorithm stops, thus it follows that $\tau \leq \ell \cdot r < \infty$ is bounded.
    Therefore, the optional stopping theorem applies to $\Phi_t + t$.
    By (e), and the optional stopping theorem
    $$\mathbb{E}[\Phi_\tau + \tau] \geq \Phi_0$$
    since $\Phi_\tau = 0$
    $$ \implies \mathbb{E}[\tau] \geq \Phi_0 = \frac{1}{2}\left( (2^{m-k}+1) \cdot \log(\ell+1) \right)$$
\end{proof}

The following is a security claim based on the established definition.

\begin{claim}
    \label{claim-secure}
    The threshold function obfuscation \\
    $\Pi = \FuncSty{preprocess},\FuncSty{probe}$ is secure as defined in \ref{definition-obfuscation-secure}.
\end{claim}

\begin{proof}
    Let $\mathcal{A}$ be any adversarial algorithm that tries to recover $T$ from $C,P,V$.
    Define $\tau,b = Obf_{\mathcal{A},\Pi,\ell}$
    By \Cref{lemma-expected-probes}, if $\mathcal{A}$ finds $T$ and does $\tau$ probes to $H_m$, it holds that
    $$\mathbb{E}[\tau] \geq \frac{1}{2}\left((2^{m-k}+1)\log (\ell + 1)\right)$$
    therefore
    $$\mathbb{E}[\tau \mid b = 1] \geq \frac{1}{2}\left((2^{m-k}+1)\log (\ell + 1)\right)$$
\end{proof}

\subsection{Implementation}
\label{section-obfuscation-implementation}

Here we give a more detailed description of the obfuscated threshold
implementation outlined in
Section~\ref{section-obfuscation-construction}.

\subsubsection{Preprocessing Algorithm} 
We present an implementation of the first algorithm of our threshold obfuscation scheme, Algorithm \ref{alg-preprocessing}. 
This algorithm accepts hash size parameter $m$, hardness parameter $k$, value range $\ell$, and input threshold $T$. 
It returns three arrays, $C$, $P$ and $V$.
The idea is to hide the threshold $T$ in an array of bits $V=v_1,\dots,v_\ell$. 
We create an array of hashes $C=C_1,\dots,C_\ell$  where each hash is created by randomly 
choosing two strings $p_i\in\{0,1\}^k$ and $r_i\in\{0,1\}^{m-k}$ for each $i=1,\dots, \ell$ and 
hashing their concatenation $p_i \Vert r_i$. We store in $C_i$ the first $m-1$ bits of the hash $H_m(p_i \Vert r_i)$ and 
the last bit we XOR with the value of $f_T(i)$ to obscure it. 
We store this in the bit $v_i \gets [H_m(p_i \Vert r_i)]_m \oplus f_T(i) $.
In this way, in order to compute $f_T(i)$, one must find a string $y$ such that $[H_m(p_i \Vert y)]_{1\dots (m-1)} = C_i$, 
and then compute $[H_m(p_i \Vert y)]_m \oplus v_i$. Given that there is a high probability 
that $[H_m(p_i \Vert y)]_{1,\dots,(m-1)}=[H_m(p_i \Vert r_i)]_{1,\dots,(m-1)} \implies [H_m(p_i \Vert y)]_m=[H_m(p_i \Vert r_i)]_m$ it 
follows that, $[H_m(p_i \Vert y)]_m \oplus v_i=[H_m(p_i \Vert y)]_m \oplus [H_m(p_i \Vert r_i)]_m \oplus f_T(i)=f_T(i)$ holds  w.h.p.
To create the hashes, we employ our random oracle function (\ref{definition-hash-function}). 
Since one needs to find a suffix for a pre-image of the hash (up to the last bit) in order to find $f_T(i)$, we must ensure that finding such a suffix is feasible. 
To control the size of the search space we have a security parameter $k$ which determines the size of the prefix string $p_i$. 
For this reason, the preprocessing algorithm returns the array of prefix strings $P=P_1,\dots, P_\ell$. 
Increasing $k$ decreases the size of the suffix required to find, which automatically reduces the search space. 
By the properties of a random oracle (\ref{definition-hash-function}), reducing the size of the search space implies reducing the runtime required to find $r_i$.

In a practical implementation of the construction, the standard hash {\cite{pub2015secure}} SHA 512/256 can be used in place of $H_m$ (\Cref{definition-hash-function}). 
While no formal proof exists, it is widely accepted that in practice, SHA behaves like a random oracle, including an avalanche effect, where any minor change in the input, even a single bit, drastically alters the output.

\begin{algorithm}[t]
    \Procedure{$\FuncSty{preprocess}(\ell, T, m, k $ s.t. $ t \leq \ell \wedge k < m)$}{
        \tcp{Initialize array of hashes and array of nonces}
        \label{line-preprocessing-declare-C}
        $C[1,...,\ell][1,...m-1]$ an uninitialized array of bit strings\;
        \label{line-preprocessing-declare-P}
        $P[1,...,\ell][1,...k]$ an uninitialized array of bit strings\;
        \label{line-preprocessing-declare-V}
        $V[1,...,\ell]$ an uninitialized array of bits\;
        \label{line-preprocessing-for-loop}
        \For{$i ← 1\dots \ell$}{
            \label{line-preprocessing-assign-P}
            $P[i] \gets$ random bit string of length $k$\;
            \label{line-preprocessing-assign-r}
            $r_i \gets$ random bit string of length $m-k$\;
            \label{line-preprocessing-assign-D}
            $d \gets H_m(P[i] \Vert r_i)$ \;
            \label{line-preprocessing-assign-C}
            $C[i] \gets d_1d_2\dots d_{m-1} $ \;
            \label{line-preprocessing-assign-V}
            $V[i] \gets d_m \oplus f_T(i)$
            \label{line-preprocessing-for-loop-end}
        }
        \Return $C,P,V$\;
        \label{line-preprocessing-return}
    }
    \caption{Threshold Encapsulation}
    \label{alg-preprocessing}
\end{algorithm}

\subsubsection{Outline} 
Algorithm~\ref{alg-preprocessing} is executed by the programmer to generate an obfuscated protocol. 
The algorithm begins by initializing the commitment array $C$, nonce
array $P$, and verification-bit array $V$
(lines~\ref{line-preprocessing-declare-C} to 
\ref{line-preprocessing-declare-V}). It then iterates over all
$\ell$ entries (lines~\ref{line-preprocessing-for-loop} to \ref{line-preprocessing-for-loop-end}). For each index $i$, the
algorithm generates a random $k$-bit nonce $P[i]$
(line~\ref{line-preprocessing-assign-P}) and a random $(m-k)$-bit string
$r_i$ (line~\ref{line-preprocessing-assign-r}). It computes $d = H_m(P[i]\Vert r_i)$
(line~\ref{line-preprocessing-assign-D}), stores the first $m-1$ bits of
$d$ in $C[i]$ (line~\ref{line-preprocessing-assign-C}), and stores the
last hash bit masked with $f_T(i)$ in $V[i]$
(line~\ref{line-preprocessing-assign-V}). Finally, the algorithm returns
the arrays $C$, $P$, and $V$ (line~\ref{line-preprocessing-return}).

\begin{algorithm}
    \LocalData {
        integers $\ell,m,k$, precomputed array of hashes $C[1,...,\ell][1,...,m-1]$, prefixes $P[1,...,\ell][1,...,k]$, and bits $V[1,\dots,\ell]$
    }
    \Procedure{$\FuncSty{probe}(i\in [\ell])$}{
        $Y [1,...,m-k]$ an uninitialized bit string\;
        $D [1,...,m]$ an uninitialized bit string\;
        \label{line-thresholdReveal-find-suffix}
        $Y \gets $ string that satisfies $C[i][1,\dots,m-1] = H_m(P[i] \Vert Y)$\;
        \label{line-thresholdReveal-assign-D}
        $D \gets H_m(P[i] \Vert Y)$\;
        \label{line-thresholdReveal-return}
        \Return $D[m] \oplus V[i]$\;
    }
    \caption{Threshold Probe}
    \label{alg-probeThreshold}
\end{algorithm}

\subsubsection{Probing algorithm}
The processes use the $\FuncSty{probe}$ procedure described
in Algorithm~\ref{alg-probeThreshold} to compute $f_T(i)$. 

The algorithm works by finding the pre-image $P[i] \Vert r_i$ that was
used to compute $C[i]$ in the preprocessing phase. It begins by
initializing a bit string $Y$ of length $m-k$ to hold potential
suffixes of the pre-image $P[i]$ of $C[i]$, and an $m$-length bit string $D$ to hold the hash that is equal to $C[i]$ in the first $m-1$ bits. Next, it finds and stores in $Y$ a
suffix for the pre-image of $C[i]$ such that the first $m-1$ bits of the hash of $P[i]\Vert Y$
are equal to $C[i]$ (line~\ref{line-thresholdReveal-find-suffix}).
It stores this hash in $D$ (line~\ref{line-thresholdReveal-assign-D}).
This means
that, with high probability $D[m] = H_m(P[i] \Vert r_i)[m]$. The exact method of finding this pre-image has not
been specified, but exhaustive search mining or any other method to
find a pre-image of a cryptographic hash function can be employed.
Thus, Algorithm~\ref{alg-probeThreshold} returns $f_T(i)$ by computing
$D[m] \oplus V[i]$ (line~\ref{line-thresholdReveal-return}), since the preprocessing algorithm defined $V[i]= H_m(P[i]\Vert r_i)[m] \oplus f_T(i)$, the algorithm returns $f_T(i)$.

\section{Conclusion}
\label{section-conclusion}

Designing an algorithm to cope with an adversarial scheduler (and/or inputs) ensures the algorithm's function in rare scenarios where the schedule is the most unfortunate for the algorithm. In many cases, the worst case leads to an impossibility result, while the rare scenario does not happen in practice. This is particularly true in scenarios that require ongoing tracing and computation of the imaginary scheduler entity.

We examine the power of using program obfuscation and random oracles
to derandomize distributed asynchronous consensus algorithms. The
power of the random oracle can be demonstrated in algorithms where
symmetry is broken by other means, such as process identifiers and/or
symmetry-breaking operations, such as compare-and-swap. 
Thus, randomization has an inherent more substantial power in breaking symmetry. On the other hand, harvesting proper randomization during runtime is challenging, and it should be avoided if possible.  

A weakness of our construction is that it requires substantial
preprocessing to construct an obfuscated threshold function since our
implementation is just a truth table hidden behind time-lock puzzles. A natural question is
whether a more sophisticated obfuscation procedure could reduce this
setup, perhaps by constructing a function $f(s,i)$ where the threshold $T$
is a hash of some shared data $s$ but the implementation of $f$
prevents recovering $T$ more efficiently than simply doing binary
search.

However, even with these limitations,
we believe that we have demonstrated that integrating program obfuscation and random oracle abstractions and functionalities in distributed computing is helpful for the theory and practice of distributed computing and systems.

\section*{Acknowledgments}
This work was supported by the Google Research Grant, the Rita Altura Trust Chair in Computer Science, the Frankel Center for Computer Science, and the Israeli Science Foundation (Grant No. 465/22).

\bibliographystyle{alpha}
\bibliography{ro4ac}

\newcommand{\etalchar}[1]{$^{#1}$}
\begin{thebibliography}{DGNW24}

\bibitem[ACH{\etalchar{+}}23]{DBLP:conf/asiacrypt/AfsharCHLM23}
Abtin Afshar, Kai{-}Min Chung, Yao{-}Ching Hsieh, Yao{-}Ting Lin, and Mohammad
  Mahmoody.
\newblock On the (im)possibility of time-lock puzzles in the quantum random
  oracle model.
\newblock In Jian Guo and Ron Steinfeld, editors, {\em ASIACRYPT 2023, Part
  IV}, volume 14441 of {\em Lecture Notes in Computer Science}, pages 339--368.
  Springer, 2023.

\bibitem[Ben83]{DBLP:conf/podc/Ben-Or83}
Michael Ben{-}Or.
\newblock Another advantage of free choice: Completely asynchronous agreement
  protocols (extended abstract).
\newblock In Robert~L. Probert, Nancy~A. Lynch, and Nicola Santoro, editors,
  {\em Proceedings of the Second Annual {ACM} Symposium on Principles of
  Distributed Computing, Montreal, Quebec, Canada, August 17-19, 1983}, pages
  27--30. {ACM}, 1983.

\bibitem[BGI{\etalchar{+}}12]{DBLP:journals/jacm/BarakGIRSVY12}
Boaz Barak, Oded Goldreich, Russell Impagliazzo, Steven Rudich, Amit Sahai,
  Salil~P. Vadhan, and Ke~Yang.
\newblock On the (im)possibility of obfuscating programs.
\newblock {\em J. {ACM}}, 59(2):6:1--6:48, 2012.

\bibitem[DGNW24]{DBLP:journals/tdsc/DolevGNW24}
Shlomi Dolev, Bingyong Guo, Jianyu Niu, and Ziyu Wang.
\newblock Sodsbc: {A} post-quantum by design asynchronous blockchain framework.
\newblock {\em {IEEE} Trans. Dependable Secur. Comput.}, 21(1):47--62, 2024.

\bibitem[FLP85]{FischerLP1985}
Michael~J. Fischer, Nancy~A. Lynch, and Mike Paterson.
\newblock Impossibility of distributed consensus with one faulty process.
\newblock {\em J. {ACM}}, 32(2):374--382, 1985.

\bibitem[FS86]{FS}
Amos Fiat and Adi Shamir.
\newblock How to prove yourself: Practical solutions to identification and
  signature problems.
\newblock In Andrew~M. Odlyzko, editor, {\em CRYPTO '86}, volume 263 of {\em
  Lecture Notes in Computer Science}, pages 186--194. Springer, 1986.

\bibitem[Gel18]{Gelashvili2018}
Rati Gelashvili.
\newblock On the optimal space complexity of consensus for anonymous processes.
\newblock {\em Distributed Computing}, 31(4):317--326, Aug 2018.

\bibitem[JLS24]{DBLP:journals/cacm/JainLS24}
Aayush Jain, Huijia Lin, and Amit Sahai.
\newblock Indistinguishability obfuscation from well-founded assumptions.
\newblock {\em Commun. {ACM}}, 67(3):97--105, 2024.

\bibitem[KL14]{KatzLindell}
Jonathan Katz and Yehuda Lindell.
\newblock {\em Introduction to Modern Cryptography, Second Edition}.
\newblock Chapman \& Hall/CRC, 2nd edition, 2014.

\bibitem[LAA87]{LouiA1987}
Michael~C. Loui and Hosame~H. Abu-Amara.
\newblock Memory requirements for agreement among unreliable asynchronous
  processes.
\newblock In Franco~P. Preparata, editor, {\em Parallel and Distributed
  Computing}, volume~4 of {\em Advances in Computing Research}, pages 163--183.
  JAI Press, 1987.

\bibitem[Lam98]{DBLP:journals/tocs/Lamport98}
Leslie Lamport.
\newblock The part-time parliament.
\newblock {\em {ACM} Trans. Comput. Syst.}, 16(2):133--169, 1998.

\bibitem[Mor95]{Moran1995}
Shlomo Moran.
\newblock Using approximate agreement to obtain complete disagreement: the
  output structure of input-free asynchronous computations.
\newblock In {\em Third Israel Symposium on the Theory of Computing and
  Systems}, pages 251--257, January 1995.

\bibitem[oST15]{pub2015secure}
National~Institute of~Standards and Technology.
\newblock Secure hash standard ({SHS}).
\newblock {\em Federal Information Processing Standards (FIPS) Publication
  180-4}, 180(4), August 2015.

\bibitem[RSW96]{rivest1996time}
Ronald~L. Rivest, Adi Shamir, and David~A. Wagner.
\newblock Time-lock puzzles and timed-release crypto.
\newblock Technical Report MIT/LCS/TR-684, MIT, February 1996.

\end{thebibliography}

\end{document}